\documentclass[12pt]{amsart}

\author{Paul \textsc{Poncet}}
\address{CMAP, \'{E}cole Polytechnique, Route de Saclay, 91128 Palaiseau Cedex, France}

\email{poncet@cmap.polytechnique.fr}

\usepackage[colorinlistoftodos,bordercolor=orange,backgroundcolor=orange!20,linecolor=orange,textsize=scriptsize]{todonotes}
\usepackage{stmaryrd}
\usepackage{amssymb}
\usepackage{amsmath}
\usepackage{graphicx}
\usepackage{t1enc}
\usepackage{color}
\usepackage[latin1]{inputenc}
\usepackage{yhmath}
\usepackage{mathabx}
\usepackage{calrsfs}                                             
\usepackage{times}
\usepackage{bbm}
\usepackage{ifpdf}

\def\twoheaddownarrow{\rlap{$\downarrow$}\raise-.5ex\hbox{$\downarrow$}}
\def\twoheaduparrow{\rlap{$\uparrow$}\raise.5ex\hbox{$\uparrow$}}

\newcommand{\cc}{\mathfrak{c}}

\newcommand{\Z}{\mathsf{Z}}

\newtheorem*{theorem*}{Theorem}

\newtheorem{theorem}{Theorem}[section]
\newtheorem{corollary}[theorem]{Corollary}
\newtheorem{proposition}[theorem]{Proposition}
\newtheorem{lemma}[theorem]{Lemma}

\theoremstyle{definition}
\newtheorem{definition}[theorem]{Definition}
\newtheorem{example}[theorem]{Example}
\newtheorem{remark}[theorem]{Remark}

\newenvironment{acknowledgements}[1][]{\par\vspace{0.5cm}\noindent\textbf{Acknowledgements#1.} }{\par}

\begin{document}

\title{Enriched closure spaces \\ as a novel framework \\ for domain theory} 
\date{\today}

\subjclass[2010]{06B35, 
				 54A05} 

\keywords{preclosure spaces, closure spaces, closure operators, continuous posets, strongly continuous posets, Z-continuous posets, interpolation property, domains, continuous lattices}

\begin{abstract}
We propose a generalization of continuous lattices and domains through the concept of enriched closure space, defined as a closure space equipped with a preclosure operator satisfying some compatibility conditions. 
In this framework we are able to define a notion of way-below relation; an appropriate definition of continuity then naturally follows.  Characterizations of continuity of the enriched closure space and necessary and sufficient conditions for the interpolation property are proved. 
We also draw a link between continuity and the possibility for the subsets that are open with respect to the preclosure operator to form a topology. 
\end{abstract}

\maketitle

\section{Introduction}

The theory of domains and continuous lattices was initiated in the early 1970's with the pioneering work of Dana S.\ Scott \cite{Scott72} and motivated by the search for a denotational semantics of lambda calculus. 
Since then it has been steadily developed and has found numerous applications in order theory, theoretical computer science, and non-Hausdorff topology -- one may refer to the monographs by Gierz et al.\ \cite{Gierz03}, Abramsky and Jung \cite{Abramsky94}, and Goubault-Larrecq \cite{Goubault13} for a deep dive.  
Additional fields of application include measure and integration (Edalat \cite{Edalat95b}, Howroyd \cite{Howroyd00}, Lawson and Lu \cite{Lawson03}), probability and large deviations (Norberg \cite{Norberg89},  \cite{Norberg97}, Jonasson \cite{Jonasson98}, Akian \cite{Akian99}), general relativity (Martin and Panangaden \cite{Martin06}), tropical convexity and idempotent mathematics (Lawson \cite{Lawson04b}, Poncet \cite{Poncet11}), to name but a few. 

The efficiency of this theory has pushed mathematicians to explore generalizations of the concepts of way-below relation and continuity. 
One fruitful proposal has consisted in replacing the family $\mathrsfs{D}$ of directed subsets with supremum, which plays a key role in classical domain theory, by some other family $\mathfrak{M}$ of subsets. 
This idea has been at the core of many interesting variants to domain theory and has brought a unified point of view on several preexisting order-theoretic results. 
Apart from $\mathrsfs{D}$, one may for instance pick $\mathfrak{M}$ among the following collections: 
\begin{itemize}
	\item the finite subsets (Martinez \cite{Martinez72}, Ern\'e \cite{Erne81}); 
	\item the totally ordered subsets with supremum 
	(Markowsky \cite{Markowsky81a}); 
	\item the subsets with supremum (Raney \cite{Raney52}, \cite{Raney53}, Ern\'e et al.\ \cite{Erne06}); 
	\item the countably directed subsets with supremum (Han et al.\ \cite{Han89});
	\item the Scott-closed subsets (see Ho and Zhao \cite{Ho09}).  
\end{itemize}
Bandelt \cite{Bandelt82}, Bandelt and Ern\'e \cite{Bandelt84}, and Xu \cite{Xu95} have dealt with the general situation of an unspecified collection $\mathfrak{M}$ fulfilling suitable properties.  
In the same line, $\mathfrak{M}$ has also been taken of the form $\Z[P]$ for a given \textit{subset system} $\Z$ as defined by Wright et al.\ \cite{Wright78}, that is a function $\Z$ assigning to each poset $P$ a collection $\Z[P]$ of subsets of $P$ with certain functorial properties; see e.g.\ Nelson \cite{Nelson81}, Bandelt and Ern\'e \cite{Bandelt83}, Venugopalan \cite{Venugopalan86}, Baranga \cite{Baranga96}, Ern\'e \cite{Erne99}. 

These general $\mathfrak{M}$ or $\Z$ frameworks 
are not fully satisfying yet, for they do not encompass certain situations encountered in the literature. 
For instance, Novak \cite{Novak82a} considered the relation defined on a poset $P$ by $x \ll y$ if $y \leqslant \ell(a)$ implies $x \leqslant r(a)$ for all $a \in Q$, for a given poset $Q$ and maps $\ell, r \colon Q \rightarrow P$ with appropriate features. 
This relation and the notion of continuity it involves cannot be depicted with the terms of $\Z$ theory. 




To overcome this hurdle, we turn towards the concept of (pre)closure operator, the role of which in domain theory has already been stressed at several occasions (see Ern\'e \cite{Erne04}, \cite{Erne09}). 
A \textit{preclosure operator} on a set $E$ is a map $\cc \colon 2^E \rightarrow 2^E$ such that $A \subset \cc(A) \subset \cc(B)$, for all subsets $A$ and $B$ such that $A \subset B$. 
A subset $A$ is then \textit{$\cc$-closed} if $\cc(A) = A$, and \textit{$\cc$-open} if its complement is $\cc$-closed. 
A \textit{closure operator} $\cc$ is a preclosure operator which is also \textit{idempotent}, in the sense that $\cc(\cc(A)) = \cc(A)$ for all subsets $A$. 
On a poset, some natural closure operators are at hand, especially the \textit{Alexandrov closure operator} 
\[
A \mapsto \downarrow\!\! A := \{ x : x \leqslant a \mbox{ for some } a \in A \}, 
\] 
and the \textit{Dedekind--MacNeille closure operator} 
\[
A \mapsto A^{\uparrow \downarrow} := \{ x : x \leqslant u \mbox{ for all upper bounds $u$ of $A$ } \}.
\] 
Note that, if $A$ admits a supremum $\bigvee A$, then $A^{\uparrow \downarrow} = \downarrow\!\! \bigvee A$.  
These two closure operators act silently in the context of completely distributive lattices: 
given a complete lattice $L$ equipped with the relation $\vartriangleleft$ defined by $x \vartriangleleft y$ if $y \leqslant \bigvee A$ implies $x \in \downarrow\!\! A$ for all subsets $A$, Raney \cite{Raney53} proved that $L$ is completely distributive if and only if, for all $x \in L$, $x = \bigvee \{ y : y \vartriangleleft x \}$, or equivalently $x \in \{ y : y \vartriangleleft x \}^{\uparrow\downarrow}$. 
Ern\'e and Wilke \cite{Erne83} generalized this situation with the concept of core space: let $E$ be a set equipped with a closure operator $\cc$, and let $\leqslant$ be its specialization quasiorder defined by $x \leqslant y$ if $x \in \cc(\{y\})$. The \textit{core relation} $\vartriangleleft$ on $E$ is defined by $x \vartriangleleft y$ if $y \in \cc(A)$ implies $x \in \downarrow\!\! A$ for all subsets $A$. Then $(E, \cc)$ is a \textit{core space} if $x \in \cc(\{ y : y \vartriangleleft x \})$ for all $x \in E$. 

In this paper we bring the use of closure and preclosure operators to a further level of abstraction. 
The central structure in our approach is a triplet $(P, [\cdot], \cc)$, that we call an \textit{enriched closet}, where $(P, [\cdot])$ is a closure space and $\cc$ is a preclosure operator on $P$ satisfying the compatibility conditions $[\cc(A)] = \cc(A) = \cc([A])$ for all subsets $A$. 
In this setting, the \textit{way-below relation} $\ll$  is defined by $x \ll y$ if $y \in \cc(A)$ implies $x \in [A]$ for all subsets $A$.  
And $(P, [\cdot], \cc)$ is \textit{continuous} if $x \in \cc(\{ y : y \ll x \})$ for all $x \in P$. 
From these definitions not only do we recover classical domain theory with $[A] := \downarrow\!\! A$ and $\cc \colon A \mapsto \bigcup_{D \subset [A]} D^{\uparrow\downarrow}$, where $D$ runs over the directed subsets with supremum included in $[A]$, but we also encompass situations left aside by the $\Z$ framework. 

Moreover, we provide a series of results; some of them generalize assertions rediscovered for different choices of $\mathfrak{M}$; some others are new, in the sense that no counterpart exists in the $\mathfrak{M}$ or $\Z$ frameworks. 
The next result is of the former type: 

\begin{theorem*}
Let $(P, [\cdot], \cc)$ be an enriched closet. 
Then the following conditions are equivalent:
\begin{itemize}
	\item $(P, [\cdot], \cc)$ is continuous; 
	\item $A \subset \cc(\twoheaddownarrow A)$, for all subsets $A$ of $P$;   
	\item $(\twoheaddownarrow, \cc)$ is a Galois connection on $[\cdot]$-closed subsets; 
	\item $\cc$ preserves arbitrary intersections of $[\cdot]$-closed subsets. 
\end{itemize}
\end{theorem*}


We also provide necessary and/or sufficient conditions for the way-below relation $\ll$ to have \textit{interpolation property} in the sense that, whenever $x \ll z$, there exists some $y$ such that $x \ll y \ll z$; equivalently, $\ll$ is idempotent for the composition of binary relations. 
Recall that, in classical domain theory, the interpolation property is satisfied as soon as the poset at stake is continuous. This implication is no longer automatic for a $\Z$-continuous poset; however, if $\Z$ is \textit{union-complete} in the sense that $\mathrsfs{V} \in \Z[\Z[P]]$ implies $\bigcup \mathrsfs{V} \in \Z[P]$, then continuity implies the interpolation property \cite{Bandelt84}. 
By transposing this definition of an adequate notion of \textit{union-complete family of subsets}, we obtain the following result. 


\begin{theorem*}
Let $(P, [\cdot], \cc)$ be a continuous enriched closet. 
Then $\ll$ has the interpolation property if and only if $\cc$ is generated by a union-complete family of relatively-closed $[\cdot]$-closed subsets. 
\end{theorem*}

Here, a preclosure operator $\cc$, is \textit{generated by a family $\mathrsfs{D}$} of subsets if for all subsets $A$ we have 
\[
\cc(A) = \bigcup_{D \subset [A]} \cc(D), 
\]
where the union is taken over the subsets $D \in \mathrsfs{D}$ included in $[A]$. 
And a subset is \textit{relatively-closed} if $\cc(\cc(A)) = \cc(A)$. 


To our knowledge, the following result has no  counterpart in the $\mathfrak{M}$ or $\Z$ frameworks. We write $\twoheaduparrow A$ for $\{ x : a \ll x \mbox{ for some $a \in A$ } \}$. 

\begin{theorem*}
Let $(P, [\cdot], \cc)$ be a continuous enriched closet in which every subset $G$ such that $G = \twoheaduparrow G$ is $[\cdot]$-open. 
Then $\ll$ has the interpolation property if and only if subsets of the form $\twoheaduparrow A$, $A \subset P$, are $\cc$-open. 
\end{theorem*}

One key aspect of our framework to be emphasized is that we do not suppose the preclosure operator $\cc$ to be idempotent in general, i.e.\ $\cc$ is not necessarily a closure operator. 
This makes explicit the correspondence between idempotency of $\cc$ and the interpolation property, as follows: 

\begin{theorem*}
Let $(P, [\cdot], \cc)$ be a continuous enriched closet whose family of $[\cdot]$-closed subsets is union-complete. 
Then $\ll$ has the interpolation property if and only if $\cc$ is a closure operator. 
\end{theorem*}

The paper is organized as follows. 
In Section~\ref{sec:preliminaries} we present or recall some elementary definitions and results regarding closure and preclosure spaces, i.e.\ sets endowed with a closure or preclosure operator. 
In Section~\ref{sec:enriched-qosets} we introduce  enriched closets. 
We demonstrate that this concept is appropriate for defining a notion of way-below relation in Section~\ref{sec:way-below}, and notions of continuity in Section~\ref{sec:approximation}, which extend the usual definitions that exist in classical continuous lattices and domain theory. We characterize continuity of enriched closets in several ways. 
In Section~\ref{sec:interpolation} we prove manifold necessary and sufficient conditions for the way-below relation to have the interpolation property. 
As a special interesting case we focus in Section~\ref{sec:inner-regular} on enriched closets whose preclosure operator is generated from below by families of subsets with special properties. 
This leads in Section~\ref{sec:topology} to conditions for the family of open subsets in an enriched closet to be a topology (which coincides with the Scott topology in the classical case). 
The subsequent developments are illustrated with a series of examples.

\section{Preliminaries on preclosure spaces}\label{sec:preliminaries}

A \textit{preclosure space} $(E, \cc)$ is a set $E$ equipped with a map $\cc \colon 2^E \rightarrow 2^E$ such that $A \subset \cc(A) \subset \cc(B)$, for all subsets $A$ and $B$ such that $A \subset B$. 
The map $\cc$ is called a \textit{preclosure operator}. 
A subset $F$ of $E$ is \textit{$\cc$-closed} (or simply \textit{closed} if the context is clear) if $\cc(F) = F$. 
A subset $G$ of $E$ is \textit{$\cc$-open} (or simply \textit{open} if the context is clear) if $E \setminus G$ is closed. 
Note that the empty set $\emptyset$ is not assumed to be closed in general. 

\begin{lemma}\label{lem:clopen}
Let $(E, \cc)$ be a preclosure space. 
A subset $F$ is closed if and only if 
\[
A \subset F \Rightarrow \cc(A) \subset F, 
\]
for all subsets $A$ of $E$. 
A subset $G$ of $E$ is open if and only if 
\[
\cc(A) \cap G \neq \emptyset \Rightarrow A \cap G \neq \emptyset,  
\]
for all subsets $A$ of $E$. 
\end{lemma}

\begin{proof}
We only prove the assertion for closed subsets. 
Assume that the condition of the lemma is satisfied for some subset $F$. Then with $A = F$ we get $\cc(F) \subset F$, so that $F = \cc(F)$, i.e.\ $F$ is closed. 
Conversely, if $F$ is a closed subset and $A \subset F$, then $\cc(A) \subset \cc(F) = F$. 
\end{proof}

\begin{proposition}
The set of closed subsets of a preclosure space is a Moore family, in the sense that it is stable under arbitrary intersections (hence is a complete lattice). 
\end{proposition}

\begin{proof}
Let $(F_j)_{j\in J}$ be a family of closed subsets, and let $F = \bigcap_{j\in J} F_j$. For all $j \in J$, we have $\cc(F) \subset \cc(F_j) = F_j$, so that $\cc(F) \subset \bigcap_{j\in J} F_j = F$, which proves that $F$ is closed. 
\end{proof}

If a preclosure operator $\cc$ on $E$ is \textit{idempotent}, i.e.\ such that $\cc(\cc(A)) = \cc(A)$ for all subsets $A$, then it is called a \textit{closure operator}, and $(E, \cc)$ is a \textit{closure space}. 
In this paper, we shall use the unconventional term of \textit{closet} for a closure space. 

The previous result shows that, if $(E, \cc)$ is a preclosure space, there exists a least closure operator $\overline{\cc}$ that dominates $\cc$ (pointwise), defined for all subsets $A$ by 
\[
\overline{\cc}(A) = \bigcap_{F \supset A} F, 
\]
where the intersection is taken over the closed subsets $F$ containing $A$. 
It follows from the definition that $\cc(A) \subset \overline{\cc}(A)$ and $\cc(\overline{\cc}(A)) = \overline{\cc}(A)$, for all subsets $A$ of $E$. 
Note also that $\cc$ and $\overline{\cc}$ define the same family of closed (resp.\ open) subsets. 
We call $\overline{\cc}$ the \textit{closure operator associated with the preclosure operator $\cc$}. 


A map $f \colon E \rightarrow E'$ between two preclosure spaces $(E,\cc)$ and $(E',\cc')$ is \textit{continuous} if $f(\cc(A)) \subset \cc'(f(A))$ for all subsets $A$ of $E$. 
If $f$ is continuous with respect to $(E, \overline{\cc})$ and $(E', \overline{\cc}')$ we shall say that $f$ is \textit{closure-continuous}, as opposed to a  \textit{strictly-continuous} map defined as a continuous map with respect to $(E, \cc)$ and $(E', \cc')$. 
It is well known that a map is closure-continuous if and only if closed subsets are preserved by the inverse image of that map. 

\begin{proposition}\label{prop:eqstrictscott}
Let $(E, \cc)$ and $(E', \cc')$ be two preclosure spaces, and let $f \colon E \rightarrow E'$. 
If $f$ is strictly-continuous, then $f$ is closure-continuous. 
Conversely, if $\cc'$ is a closure operator (i.e.\ $\cc' = \overline{\cc}'$) and $f$ is closure-continuous, then $f$ is strictly-continuous. 
\end{proposition}

\begin{proof}
Let $F'$ be a closed subset of $E'$, and let us show that $F := f^{-1}(F')$ is closed in $E$. 
So let $A \subset F$. 
Then $f(A) \subset F'$, hence $f(\cc(A)) \subset \cc'(f(A)) \subset \cc'(F') = F'$. 
This gives $\cc(A) \subset F$. 
So by Lemma~\ref{lem:clopen} $F$ is closed. 
The inverse image of every closed subset is closed, so $f$ is closure-continuous. 

Conversely, assume that $f$ is closure-continuous and $\cc'$ is a closure operator, and let $A$ be a subset of $E$. 
By closure-continuity $f(\cc(A)) \subset f(\overline{\cc}(A)) \subset \overline{\cc}'(f(A))$. 
Since $\cc'$ is a closure operator, $\overline{\cc}'(f(A)) = \cc'(f(A))$. So we have $f(\cc(A)) \subset \cc'(f(A))$ for all subsets $A$, which proves that $f$ is strictly-continuous. 
\end{proof}


For more on closure spaces, see e.g.\ Ern\'e \cite{Erne09}.

\section{The notion of enriched closet}\label{sec:enriched-qosets}

\begin{definition}\label{def:ec}
An \textit{enriched closet} is a closet $(P, [\cdot])$ equipped with a \textit{compatible} preclosure operator $\cc$, in the sense that 
\[
[\cc(A)] = \cc(A) = \cc([A]), 
\] 
for all subsets $A$ of $P$.  
\end{definition}

Open (resp.\ closed) subsets will always refer to $\cc$-open (resp.\ $\cc$-closed) subsets, while a $[\cdot]$-open (resp.\ $[\cdot]$-closed) subset will be called \textit{weakly-open} (resp.\ \textit{weakly-closed}). 
Note that: 
\begin{itemize}
	\item every open subset is weakly-open; 
	\item every closed subset is weakly-closed; 
	\item every subset of the form $\cc(A)$ is weakly-closed. 
\end{itemize}

A \textit{quasiordered set} or \textit{qoset} $(P,\leqslant)$ is a set $P$ together with a reflexive and transitive binary relation $\leqslant$. 
If in addition $\leqslant$ is antisymmetric, then $(P, \leqslant)$ is a \textit{partially ordered set} or \textit{poset}. 
The sets $\uparrow\!\! A := \{ x \in P : \exists a\in A, a \leqslant x \}$ and $\downarrow\!\! A := \{ x \in P : \exists a \in A, x \leqslant a \}$ are the \textit{upper set} and the \textit{lower set} generated by $A$, respectively. 
If $x\in P$, we write $\uparrow\!\! x$ for the principal filter $\uparrow\!\! \{ x \}$, and $\downarrow\!\! x$ for the principal ideal $\downarrow\!\! \{ x \}$. 
The subset $A^{\downarrow} = \{ x \in P : \forall a \in A, x \leqslant a \}$ is the set of lower bounds of $A$. The set $A^{\uparrow}$ of upper bounds of $A$ is defined dually. 
A nonempty subset $D$ of $P$ is \textit{directed} if for all $d, d' \in D$, there exists some $d'' \in D$ such that $d \leqslant d''$ and $d' \leqslant d''$.

\begin{definition}\label{def:ep}
An \textit{enriched qoset} (resp.\ \textit{enriched poset}) is a qoset (resp.\ poset) $(P, \leqslant)$ equipped with a preclosure operator $\cc$ making $(P, [\cdot], \cc)$ into an enriched closet, where $[A] := \downarrow\!\! A$. 
In this situation $\cc$ is said to be \textit{compatible} with the quasiorder $\leqslant$. 
\end{definition}

\begin{example}\label{ex:alex}
Let $(P,\leqslant)$ be a qoset. 
Several (pre)closure operators that are compatible with $\leqslant$ coexist on $P$. 
Especially, the \textit{Alexandrov closure operator} is $A \mapsto \downarrow\!\! A$, 
and the \textit{Dedekind--MacNeille closure operator} is $A \mapsto A^{\uparrow \downarrow} := (A^{\uparrow})^{\downarrow}$. If $P$ is a poset and $A$ admits a least upper bound $\bigvee A$, then $A^{\uparrow \downarrow} = \downarrow\!\! \bigvee A$. 
\end{example}

\begin{example}\label{ex:p2}
Let $(P, \leqslant)$ be a poset, and let $\mathrsfs{D}$ be the family of directed subsets with supremum. Consider the map $\cc_{\mathrsfs{D}}$ defined for all subsets $A$ of $P$ by 
\[
\cc_{\mathrsfs{D}}(A) = \bigcup_{D \subset \downarrow\!\! A} D^{\uparrow\downarrow}, 
\]
where the union is taken over the elements $D$ of $\mathrsfs{D}$ that are included in $\downarrow\!\! A$. 
Then $\cc_{\mathrsfs{D}}$ is a preclosure operator that is compatible with $\leqslant$. 
Moreover, a subset $F$ is closed with respect to $\cc_{\mathrsfs{D}}$ if and only if it is closed in the Scott topology (see e.g.\ \cite  [Definition~II-1.3]{Gierz03}). This means that the closure operator $\overline{\cc}_{\mathrsfs{D}}$ associated with $\cc_{\mathrsfs{D}}$ coincides with the closure operator associated with the Scott topology. 
We shall examine in Section~\ref{sec:topology} necessary and sufficient conditions for a (pre)closure operator to be \textit{topological}, as is the case in this example. 
We shall see in Section~\ref{sec:approximation} how $\cc_{\mathrsfs{D}}$ enables one to recover the usual notion of continuous poset from the abstract notion of continuity developed below. 
\end{example}

\begin{example}\label{ex:c2}
Let $(P, \leqslant)$ be a qoset, and $c \colon P \rightarrow P$ be an order-preserving and inflationary map, i.e.\ satisfying $x \leqslant c(x) \leqslant c(y)$ for all $x, y \in P$ with $x \leqslant y$. Then the map $\cc_c \colon A \mapsto \downarrow\!\! \{ c(x) : x \in \downarrow\!\! A \}$ is a preclosure operator that is compatible with $\leqslant$. 
Moreover, $\cc_c$ is a closure operator if and only if $c \circ c = c$. 
\end{example}

\begin{example}\protect{[Novak \cite{Novak82a}]}\label{ex:novak2}
As a generalization of the previous example, let $P$ and $Q$ be qosets, and let $(\ell, r)$ be a pair of order-preserving maps with $\ell \colon Q \rightarrow P$ and $r \colon P \rightarrow Q$ such that $x \leqslant \ell(r(x))$ for all $x \in P$. 
Then the map $\cc_{\ell, r} \colon A \mapsto \downarrow\!\! \{ \ell(y) : y \in \downarrow\!\! r(A) \}$ is a preclosure operator that is compatible with $\leqslant$. 
Moreover, $\cc_{\ell, r}$ is a closure operator if and only if $\ell \circ r \circ \ell = \ell$. 
\end{example}

\begin{example}\label{ex:phi2}
Let $(P, \leqslant)$ be a qoset and $\Phi$ be a family of order-preserving self-maps on $P$. We suppose that $\Phi$ contains at least one deflationary map (i.e.\ a map $\varphi$ such that $\varphi(x) \leqslant x$ for all $x$). 
Then the map $A \mapsto \cc_{\Phi}(A) = \{ x \in P : \varphi(x) \in \downarrow\!\! A \mbox{ for some } \varphi \in \Phi \}$ is a preclosure operator that is compatible with the quasiorder $\leqslant$. 
If moreover $\Phi$ is stable by composition, then $\cc_{\Phi}$ is a closure operator. 
\end{example}

\section{The way-below relation}\label{sec:way-below}

We introduce the way-below relation in the setting of enriched closets. 
In the next section this will enable us to define a notion of continuity that encompasses the usual notion of poset continuity. 
In what follows, $(P,[\cdot], \cc)$ denotes an enriched closet.  

If $x, y \in P$, then $x$ is \textit{way-below} $y$, in symbols $x \ll y$, if 
\[
y \in \cc(A) \Rightarrow x \in [A], 
\] 
for all subsets $A$ of $P$. 
We write $\twoheaduparrow A := \{ x \in P : a \ll x \mbox{ for some } a \in A \}$ and we define $\twoheaddownarrow A$ dually. We write $\twoheaddownarrow x$ for $\twoheaddownarrow \{x\}$ and $\twoheaduparrow x$ for $\twoheaduparrow \{x\}$. 

\begin{proposition}
Let $(P, [\cdot], \cc)$ be an enriched closet. Then 
\begin{itemize}
	\item the way-below relation is transitive;
	\item $x \ll y$ implies $x \leqslant y$;
	\item $x \leqslant y \ll z$ implies $x \ll z$; 
	\item $x \ll y \leqslant z$ implies $x \ll z$; 
	\item $\twoheaddownarrow x \subset \downarrow\!\! x \subset \cc(\{x\})$, 
\end{itemize}
where $\leqslant$ denotes the quasiorder on $P$ defined by $x \leqslant y$ if $x \in [\{y\}]$. 
\end{proposition}

\begin{proof}
Suppose that $x \ll y$ and $y \ll z$, and let $A$ be a subset such that $z \in \cc(A)$. 
Then $y \in [A] \subset \cc(A)$, hence $x \in [A]$. 
This proves that $x \ll z$. So $\ll$ is transitive. 

Suppose that $x \ll y$. From $y \in \cc(\{y\})$ we get $x \in [\{y\}]$, i.e.\ $x \leqslant y$. 

Suppose that $x \leqslant y \ll z$, and let $A$ be a subset such that $z \in \cc(A)$. Then $y \in [A]$, so $x \in [\{y\}] \subset [[A]] = [A]$. This proves that $x \ll z$. 

Suppose that $x \ll y \leqslant z$, and let $A$ be a subset such that $z \in \cc(A)$. 
Then $y \in [\{z\}] \subset [\cc(A)] = \cc(A)$, so $x \in [A]$. This proves that $x \ll z$. 

Once recalled that $\downarrow\!\! x = [\{ x \}]$, inclusions $\twoheaddownarrow x \subset \downarrow\!\! x \subset \cc(\{x\})$ become straightforward. 
\end{proof}

\begin{proposition}\label{prop:scott0}
Let $(P, [\cdot], \cc)$ be an enriched closet. 
If $G$ is a weakly-open subset such that $G = \twoheaduparrow G$, then $G$ is open. 
\end{proposition}

\begin{proof}
Assume that $G$ is weakly-open and such that $G = \twoheaduparrow G$, and let $A$ such that $\cc(A) \cap G \neq \emptyset$. 
There is some $x \in \cc(A)$ with $x \in G$. Then $x \in \twoheaduparrow G$, meaning that there is some $y \in G$ such that $y \ll x$. 
This implies that $y \in [A]$. 
If $A \cap G = \emptyset$, then $A \subset P \setminus G$, so $y \in [A] \subset [P \setminus G] = P \setminus G$, a contradiction. 
This shows that $A \cap G \neq \emptyset$. 
By Lemma~\ref{lem:clopen}, this proves that $G$ is open. 
\end{proof}

We call \textit{way-upper} a subset $G$ such that $G = \twoheaduparrow G$. So the previous proposition says that every weakly-open and way-upper subset is open. Proposition~\ref{prop:scott} below will give a converse statement in the case where $P$ is continuous. 

\begin{example}[Example~\ref{ex:p2} continued]\label{ex:p3}
Let $(P, \leqslant)$ be a poset, and let $\mathrsfs{D}$ be the family of directed subsets with supremum. 
Then the way-below relation associated with the compatible preclosure operator $\cc_{\mathrsfs{D}}$ coincides with the usual way-below relation of the poset $P$. In other words, $x \ll y$ if and only if, whenever $y \leqslant \bigvee D$ for some directed subset $D$ with supremum, we have $x \in \downarrow\!\! D$. 
%
\end{example}

\begin{example}[Example~\ref{ex:c2} continued]\label{ex:c3}
Let $(P, \leqslant)$ be a qoset, and $c \colon P \rightarrow P$ be an order-preserving and inflationary map. 
The way-below relation associated with the compatible preclosure operator $\cc_c$ is such that $x \ll y$ if and only if $y \leqslant c(a)$ implies $x \leqslant a$ for all $a \in P$. 
We have seen that $\cc_c$ is a closure operator if and only if $c \circ c = c$; if $P$ is a complete lattice, this latter condition is equivalent to $k \circ k = k$, where $k(x) := \bigvee \twoheaddownarrow x$. 
\end{example}

\begin{example}[Example~\ref{ex:novak2} continued]\label{ex:novak3}
Let $P$ and $Q$ be qosets, and let $(\ell, r)$ be a pair of order-preserving maps with $\ell \colon Q \rightarrow P$ and $r \colon P \rightarrow Q$ having the following properties: 
\begin{itemize}
	\item $x \leqslant \ell(r(x))$ for all $x \in P$; 
	\item $r(x) \leqslant r(x') \Rightarrow x \leqslant x'$ for all $x, x' \in P$; 
	\item either $y \geqslant r(\ell(y))$ for all $y \in Q$, or $r$ is surjective.  
\end{itemize}
The way-below relation associated with the compatible preclosure operator $\cc_{\ell, r}$ is such that $x \ll x'$ if and only if $x' \leqslant \ell(a)$ implies $r(x) \leqslant a$ for all $a \in Q$. 
\end{example}

\begin{example}[Example~\ref{ex:phi2} continued]\label{ex:phi3}
Let $(P, \leqslant)$ be a qoset, and let $\Phi$ be a family of order-preserving self-maps on $P$ containing at least one deflationary map. 
The way-below relation associated with the compatible preclosure operator $\cc_{\Phi}$ is such that $x \ll y$ if and only if $x \leqslant \varphi(y)$ for all $\varphi \in \Phi$. 
Indeed, first assume that $x \ll y$, and let $\varphi \in \Phi$. Defining $A = \downarrow\!\! \varphi(y)$, we have $y \in \cc_{\Phi}(A)$, so that $x \in \downarrow\!\! A$, i.e.\ $x \leqslant \varphi(y)$. 
Conversely, assume that $x \leqslant \varphi(y)$ for all $\varphi \in \Phi$, and let $A \subset P$ such that $y \in \cc_{\Phi}(A)$. Then $\varphi_0(y) \leqslant a_0$ for some $\varphi_0 \in \Phi$ and $a_0 \in A$, so that $x \leqslant a_0$, hence $x \in \downarrow\!\! A$. This shows that $x \ll y$. 
\end{example}

\begin{example}\label{ex:hyp2}
Let $(P,\leqslant)$ be a qoset. 
The preclosure operator considered here is the topological closure with respect to the upper topology (thus, the quasiorder $\leqslant$ and the specialization agree). Then $\ll$ coincides with the \textit{hyper way-below relation} of $P$, and $x \ll y$ if and only if $y$ is in the interior of $\uparrow\!\! x$. 
See e.g.\ Mao and Xu \cite{Mao06}. 
\end{example}

\begin{example}
If $(E,\cc)$ is a closure space, and $\leqslant$ its specialization quasiorder defined by $x \leqslant y$ if $x \in \cc(\{y\})$, then the way-below relation $\ll$ is the so-called \textit{core relation} or \textit{interior relation} $\rho$ on $E$. See Ern\'e and Wilke \cite{Erne83}. 
\end{example}

\section{Approximation and continuity}\label{sec:approximation}


Hereunder we develop a notion of approximation for enriched closets similar to the usual notion of approximation on posets. 
We say that the enriched closet $(P, [\cdot], \cc)$ is \textit{continuous at $x \in P$} if $x \in \cc(\twoheaddownarrow x)$, and that $P$ is \textit{continuous} if it is continuous at every $x \in P$. 

Before giving several characterizations of continuity, recall first that a \textit{Galois connection} between two qosets $(P,\leqslant)$ and $(P', \leqslant)$ is a pair $(\phi, \psi)$ of maps $\phi \colon P \rightarrow P'$ and $\psi \colon P' \rightarrow P$ such that $\phi(x) \leqslant x'$ if and only if $x \leqslant \psi(x')$ for all $x \in P$, $x' \in P'$. The map $\phi$ (resp.\ $\psi$) is the \textit{left map} (resp.\ \textit{right map}) of the Galois connection. 

The following theorem is a generalization of Ern\'e and Wilke \cite[Theorem~2.5]{Erne83} and Ern\'e \cite[Theorem~1]{Erne81}. 

\begin{theorem}\label{thm:global}
Let $(P, [\cdot], \cc)$ be an enriched closet. Then the following conditions are equivalent:
\begin{enumerate}
	\item\label{thm:global0} $(P, [\cdot], \cc)$ is continuous; 
	\item\label{thm:global1} $A \subset \cc(\twoheaddownarrow A)$, for all subsets $A$ of $P$;   
	\item\label{thm:global3} $(\twoheaddownarrow, \cc)$ is a Galois connection on weakly-closed subsets; 
	\item\label{thm:global4} $\cc$ preserves arbitrary intersections of weakly-closed subsets. 
\end{enumerate}
\end{theorem}

\begin{proof}
\eqref{thm:global0} $\Rightarrow$ \eqref{thm:global1}. Let $A \subset P$. Then $x \in \cc(\twoheaddownarrow x) \subset \cc(\twoheaddownarrow A)$ for all $x \in A$, so that $A \subset \cc(\twoheaddownarrow A)$. 

\eqref{thm:global1} $\Rightarrow$ \eqref{thm:global3}. We want to show the equivalence $\twoheaddownarrow A \subset B \Leftrightarrow A \subset \cc(B)$ for weakly-closed subsets $A$ and $B$. 
First assume that $\twoheaddownarrow A \subset B$, and let $x \in A$. Then $\twoheaddownarrow x \subset \twoheaddownarrow A \subset B$, hence $x \in \cc(\twoheaddownarrow x) \subset \cc(B)$. 
This shows that $A \subset \cc(B)$. 
Conversely, assume that $A \subset \cc(B)$, and let $x \in \twoheaddownarrow A$.  
There is some $y \in A$ such that $x \ll y$. Then $y \in \cc(B)$, hence $x \in [B] = B$. This proves that $\twoheaddownarrow A \subset B$. 

\eqref{thm:global3} $\Rightarrow$ \eqref{thm:global4}. It comes from general properties of Galois connections. 

\eqref{thm:global4} $\Rightarrow$ \eqref{thm:global0}. If $x\in P$, then 
\[
\cc(\twoheaddownarrow x) = \cc(\bigcap_{x \in \cc(A)} [A]) = \bigcap_{x \in \cc(A)} \cc([A]) \ni x, 
\]
so that $(P, [\cdot], \cc)$ is continuous at $x$. 
\end{proof}


\begin{example}[Example~\ref{ex:p3} continued]\label{ex:p4}
Let $(P, \leqslant)$ be a poset, and let $\mathrsfs{D}$ be the family of directed subsets with supremum. 
Then $(P, \leqslant, \cc_{\mathrsfs{D}})$ is continuous if and only if every $x \in P$ is the directed supremum of elements way-below it. 
Thus we recover the usual notion of continuous poset. 
\end{example}

\begin{example}[Example~\ref{ex:c3} continued]\label{ex:c4}
Let $(P, \leqslant)$ be a qoset, and $c \colon P \rightarrow P$ be an order-preserving and inflationary map. 
Then $(P, \leqslant, \cc_c)$ is continuous if and only if, for all $x \in P$, there is some $y \in P$ such that $y \ll x \leqslant c(y)$. 
In the case where $P$ is a complete lattice, $(P, \leqslant, \cc_c)$ is continuous if and only if $(k, c)$ is a Galois connection and $k(x) \ll x$ for all $x \in P$, where $k(x) := \bigvee \twoheaddownarrow x$. 
\end{example}

\begin{example}[Example~\ref{ex:novak3} continued]\label{ex:novak4}
Let $P$ and $Q$ be qosets, and let $(\ell, r)$ be a pair of order-preserving maps with $\ell \colon Q \rightarrow P$ and $r \colon P \rightarrow Q$ having the following properties: 
\begin{itemize}
	\item $x \leqslant \ell(r(x))$ for all $x \in P$; 
	\item $r(x) \leqslant r(x') \Rightarrow x \leqslant x'$ for all $x, x' \in P$; 
	\item either $y \geqslant r(\ell(y))$ for all $y \in Q$, or $r$ is surjective.  
\end{itemize}
Then $(P, \leqslant, \cc_{\ell, r})$ is continuous if and only if, for all $x \in P$, there is some $x' \in P$ such that $x' \ll x \leqslant \ell(r(x'))$. 
In the case where $P$ is a complete lattice, $(P, \leqslant, \cc_{\ell, r})$ is continuous if and only if $(k, \ell \circ r)$ is a Galois connection and $k(x) \ll x$ for all $x \in P$, where $k(x) := \bigvee \twoheaddownarrow x$. 
\end{example}

\begin{example}[Example~\ref{ex:phi3} continued]\label{ex:phi4}
Let $(P, \leqslant)$ be a qoset, and let $\Phi$ be a family of order-preserving self-maps on $P$ containing at least one deflationary map. 
Then $(P, \leqslant, \cc_{\Phi})$ is continuous if and only if $\{ \varphi(x) : \varphi \in \Phi \}$ has a least element for all $x \in P$. 
\end{example}

We have seen with Proposition~\ref{prop:scott0} that if $G$ is a weakly-open and way-upper subset, then $G$ is open. 
Here is a converse statement in the case where $P$ is continuous, which generalizes \cite[Proposition~II-1.10(i)]{Gierz03}. 

\begin{proposition}\label{prop:scott}
Let $(P, \leqslant, \cc)$ be a continuous enriched closet. 
Then a subset $G$ is open if and only if $G$ is weakly-open and way-upper. 
\end{proposition}

\begin{proof}
If $G$ is weakly-open and $G = \twoheaduparrow G$, then $G$ is open by Proposition~\ref{prop:scott0}. 
Conversely, suppose that $G$ is open. Then $G$ is  weakly-open. 
Moreover, if $x \in P \setminus \twoheaduparrow G$, then for all $y \ll x$, $y \in P\setminus G$, i.e.\ $\twoheaddownarrow x \subset P\setminus G$. 
Since $P$ is continuous and $P\setminus G$ is closed we have $x \in \cc(\twoheaddownarrow x) \subset P\setminus G$. 
This proves that $G \subset \twoheaduparrow G$. 
The reverse inclusion is true since $\twoheaduparrow G \subset \uparrow\!\! G = G$. 
\end{proof}

A subset $A$ is called \textit{connected} if, whenever $A \subset G \cup G'$ for some disjoint open subsets $G, G'$, then $A \subset G$ or $A \subset G'$. 

\begin{corollary}
Let $(P, [\cdot], \cc)$ be a continuous enriched closet. 
Then every subset of the form $\twoheaddownarrow x$, with $x \in P$, is connected. 
\end{corollary}

\begin{proof}
Assume that $\twoheaddownarrow x \subset G \cup G'$ for some disjoint open subsets $G, G'$. 
If $\twoheaddownarrow x \not\subset G$ and $\twoheaddownarrow x \not\subset G'$, there are some $y \ll x$ and $y' \ll x$ such that $y \notin G$ and $y' \notin G'$. 
Thus, $y \in G'$ and $y' \in G$, so that $x \in \twoheaduparrow G' \cap \twoheaduparrow G = G' \cap G = \emptyset$, a contradiction. 
\end{proof}

Let $(P, [\cdot], \cc)$ and $(P', [\cdot], \cc')$ be enriched closets. 
A map $\phi \colon P \rightarrow P'$ is said to \textit{preserve the way-below relation} if $\phi(x) \ll \phi(y)$ whenever $x \ll y$. 
Recall that a map $\psi \colon P' \rightarrow P$ is said to be \textit{strictly-continuous} if $\psi(\cc'(A')) \subset \cc(\psi(A'))$, for all $A' \subset P'$. 
The following proposition extends the Bandelt--Ern\'{e} lemma \cite{Bandelt83}. 

\begin{proposition}
Let $(P, [\cdot], \cc)$ and $(P', [\cdot], \cc')$ be enriched closets, and let $\phi \colon P \rightarrow P'$ and $\psi \colon P' \rightarrow P$ be maps satisfying 
\[
\phi^{-1}([A']) = [\psi(A')] 
\]
for all $A' \subset P'$. 
If $\psi$ is strictly-continuous, then $\phi$ preserves the way-below relation. 
Conversely, if $P$ is continuous and $\phi$ preserves the way-below relation, then $\psi$ is strictly-continuous. 
\end{proposition}


\begin{proof}
Assume first that $\psi$ is strictly-continuous. 
Let $x, y \in P$ with $x \ll y$. 
We want to show that $\phi(x) \ll \phi(y)$. 
So suppose that $\phi(y) \in \cc'(A')$ for some $A' \subset P'$. 
Then $y \in \phi^{-1}(\cc'(A')) = \phi^{-1}([\cc'(A')]) = [\psi(\cc'(A'))]$. 
From the strict-continuity of $\psi$ we deduce that $y \in [\cc(\psi(A'))] = \cc(\psi(A'))$. 
With $x \ll y$ it leads to $x \in [\psi(A')] = \phi^{-1}([A'])$. So $\phi(x) \in [A']$. 
Therefore $\phi(x) \ll \phi(y)$, and thus $\phi$ preserves the way-below relation. 

Conversely, assume that $P$ is continuous and $\phi$ preserves the way-below relation. We show that, if $A' \subset P'$, then $\psi(\cc'(A')) \subset \cc(\psi(A'))$. 
So let $x \in \psi(\cc'(A'))$. 
There is some $x' \in \cc'(A')$ with $x = \psi(x')$. 
Then $\twoheaddownarrow x' \subset [A']$, and since $\phi$ preserves the way-below relation we get $\phi(\twoheaddownarrow x) \subset \twoheaddownarrow \phi(x)$. 
Moreover, $\psi(x') \in [\psi(\{ x' \})] = \phi^{-1}([\{ x' \}])$ entails $\phi(x) = \phi(\psi(x')) \leqslant x'$, so $\twoheaddownarrow \phi(x) \subset \twoheaddownarrow x'$. 
Putting things together this leads to $\phi(\twoheaddownarrow x) \subset [A']$. 
This implies that $\twoheaddownarrow x \subset \phi^{-1}([A']) = [\psi(A')]$. 
Now $P$ is continuous by hypothesis, so $x \in \cc(\twoheaddownarrow x) \subset \cc([\psi(A')]) = \cc(\psi(A'))$. 
This shows that $\psi$ is strictly-continuous. 
\end{proof}

\begin{remark}
In the special case where $P$ and $P'$ are enriched qosets, i.e.\ $[\cdot] = \downarrow\!\! \cdot$, then a pair of maps $(\phi, \psi)$ satisfies $\phi^{-1}([A']) = [\psi(A')]$ for all $A' \subset P'$ if and only if it is a Galois connection.  
\end{remark}

%

If $(P, [\cdot], \cc)$ is an enriched closet, a subset $B$ is a \textit{basis} if 
$x \in \cc(\twoheaddownarrow x \cap B)$ for all $x \in P$. 
Note that $P$ is continuous if and only if $P$ is a basis. 
The following result generalizes \cite[Proposition~III-4.2]{Gierz03}. 

\begin{proposition}
Let $(P, [\cdot], \cc)$ be an enriched closet. 
Then $P$ has a basis $B$ if and only if $P$ is continuous and for all $x, y \in P$, $x \ll y$ implies $x \ll_B y$, where $\ll_B$ is the relation defined by $x \ll_B y$ if $x \in [\twoheaddownarrow y \cap B]$. 
\end{proposition}

\begin{proof}
Assume first that $P$ has a basis $B$. Then $P$ is obviously continuous. 
If $x \ll y$, then $x \in \twoheaddownarrow y \subset  [\twoheaddownarrow y \cap B]$, so $x \ll_B y$ by definition of $\ll_B$. 

Conversely, assume that $P$ is continuous and that for all $x, y \in P$, $x \ll y$ implies $x \ll_B y$. 
Then, for all $y \in P$, $\twoheaddownarrow y 
\subset [\twoheaddownarrow y \cap B]$. 
So by continuity $y \in \cc(\twoheaddownarrow y) \subset \cc([\twoheaddownarrow y \cap B]) = \cc(\twoheaddownarrow y \cap B)$. 
Thus, $B$ is a basis of $P$. 
\end{proof}

We say that the enriched closet $(P, [\cdot], \cc)$ is \textit{algebraic} if the subset of compact elements is a basis, where an element $k$ is \textit{compact} if $k \ll k$. 
The previous proposition implies that every algebraic enriched closet is continuous. 

\begin{example}
Let $(P, \leqslant)$ be a qoset, and let $K$ be a nonempty subset of $P$. 
Then the map $\cc_K \colon A \mapsto \{ x \in P : \downarrow\!\! x \subset \downarrow\!\! A \cup (P \setminus K) \}$ is a closure operator that is compatible with $\leqslant$. 
Moreover, $x \ll y$ if and only if $x \leqslant k \leqslant y$ for some $k \in K$, the compact elements are exactly the elements of $K$, and $(P, \leqslant, \cc_K)$ is algebraic. 
\end{example}

\section{The interpolation property}\label{sec:interpolation}

In this section, we examine necessary and/or sufficient conditions related to the interpolation property. 
Some conditions we derive do not appear in the existing literature on continuous posets. 
Given an enriched closet $(P, [\cdot], \cc)$, we say that the way-below relation $\ll$ is \textit{interpolating} (or that $P$ has the \textit{interpolation property})  if $\ll$ is idempotent as a binary relation, i.e.\ if, whenever $x \ll z$, there exists some $y$ such that $x \ll y \ll z$. 
Equivalently, $\twoheaduparrow x = \twoheaduparrow \twoheaduparrow x$ (resp.\ $\twoheaddownarrow x = \twoheaddownarrow \twoheaddownarrow x$) for all $x$. 
If $P$ is continuous and $\ll$ is interpolating, we say that $P$ is \textit{strongly continuous} (in that regard we adopt the wording used in the context of $\Z$-theory, see \cite{Novak82b},  \cite{Bandelt83}, \cite{Venugopalan86}). 

Recall that, in a closet $(P, [\cdot])$, the relation $\leqslant$ is the quasiorder defined by $x \leqslant y$ if $x \in [\{y\}]$, and that $\uparrow\!\! x$ denotes the subset $\{ y \in P : x \leqslant y \}$. 

\begin{proposition}\label{lem:interp}
Let $(P, [\cdot], \cc)$ be an enriched closet with the interpolation property. 
Then $x \ll y$ if and only if $y \in G \subset \uparrow\!\! x$ for some way-upper subset $G$. 
\end{proposition}

\begin{proof}
If $x \ll y$, then $y \in G \subset \uparrow\!\! x$ with $G = \twoheaduparrow x$, and $G$ is way-upper by the interpolation property. 
Conversely, assume that $y \in G \subset \uparrow\!\! x$ for some way-upper subset $G$, and let $A$ be a subset such that $y \in \cc(A)$. 
There exists some $z \in G$ such that $z \ll y$, so $z \in [A]$. Since $z \in G \subset \uparrow\!\! x$, we get $x \in [\{z\}] \subset [A]$. 
This shows that $x \ll y$. 
\end{proof}


\begin{proposition}
Let $(P, [\cdot], \cc)$ be a continuous enriched closet. 
Then $P$ is strongly continuous if and only if $\twoheaddownarrow \cc(\cdot) = \twoheaddownarrow [\cdot]$. 
\end{proposition}

\begin{proof}
Assume first that $\ll$ is interpolating, and let $A$ be a subset of $P$. Then $\twoheaddownarrow \cc(A) = \twoheaddownarrow \twoheaddownarrow \cc(A) = \twoheaddownarrow [A]$. 
Conversely, assume that $\twoheaddownarrow \cc(\cdot) = \twoheaddownarrow [\cdot]$, and let $x \in P$. 
By continuity, $x \in \cc(\twoheaddownarrow x)$, so $\twoheaddownarrow x \subset \twoheaddownarrow \cc(\twoheaddownarrow x) = \twoheaddownarrow [\twoheaddownarrow x] = \twoheaddownarrow \twoheaddownarrow x \subset \twoheaddownarrow x$. 
This proves that $\twoheaddownarrow x = \twoheaddownarrow \twoheaddownarrow x$ for all $x$, i.e.\ $\ll$ is interpolating. 
\end{proof}

The following theorem extends Ern\'e \cite[Theorem~2]{Erne81}. 

\begin{theorem}\label{thm:interp0}
Let $(P, [\cdot], \cc)$ be a continuous enriched closet in which every way-upper subset is weakly-open. 
Then $P$ is strongly continuous, if and only if subsets of the form $\twoheaduparrow A$, $A \subset P$, are open. \end{theorem}

\begin{proof}
Assume first that $\ll$ is interpolating, and let $G = \twoheaduparrow A$. 
The interpolation property implies that $\twoheaduparrow G = \twoheaduparrow \twoheaduparrow A = \twoheaduparrow A = G$, so $G$ is way-upper, hence also weakly-open by hypothesis. 
By Proposition~\ref{prop:scott0} we deduce that $G$ is open. 

Conversely, assume that subsets of the form $\twoheaduparrow A$ are open. 
To show that $\ll$ is interpolating we prove that $\twoheaduparrow \twoheaduparrow x = \twoheaduparrow x$ for every $x$. 
So let $x \in P$, and let $G = \twoheaduparrow x$. 
By hypothesis $G$ is open, and since $P$ is continuous Proposition~\ref{prop:scott} applies and gives $\twoheaduparrow G = G$, i.e.\ $\twoheaduparrow \twoheaduparrow x = \twoheaduparrow x$. 
\end{proof}

\begin{remark}
In an enriched qoset, the condition that every way-upper subset be weakly-open always holds. 
\end{remark}

\begin{remark}
The proof also shows that subsets of the form $\twoheaduparrow A$, $A \subset P$, can be replaced by subsets of the form $\twoheaduparrow x$, $x \in P$, in the statement of the theorem. 
\end{remark}

Let $(P, [\cdot], \cc)$ be an enriched closet. 
A subset $D$ is \textit{relatively-closed} if $\cc(D)$ is closed, i.e.\ if $\cc(\cc(D)) = \cc(D)$, or equivalently if $\overline{\cc}(D) = \cc(D)$. 

\begin{lemma}\label{lem:interp1}
Let $(P, [\cdot], \cc)$ be a continuous enriched closet, and let $x, y \in P$. 
If $x \ll y$ and $\twoheaddownarrow\twoheaddownarrow y$ is weakly-closed and relatively-closed, then $x \ll a \ll y$ for some $a \in P$. 
\end{lemma}

\begin{proof}
Let $A$ be the subset $A = \twoheaddownarrow y$. 
Since $P$ is continuous we have $y \in \cc(A)$ and $A \subset \cc(\twoheaddownarrow A)$. 
This implies $y \in \cc(\cc(\twoheaddownarrow A))$. 
Since $\twoheaddownarrow A$ is relatively-closed by assumption, we get $y \in \cc(\twoheaddownarrow A)$. 
Now $x \ll y$, so $x \in [\twoheaddownarrow A]$. 
Since $\twoheaddownarrow A$ is weakly-closed by assumption, we obtain $x \in \twoheaddownarrow A$. 
Consequently, $x \ll a$ for some $a \in A = \twoheaddownarrow y$, i.e.\ $x \ll a \ll y$ for some $a \in P$. 
\end{proof}

In a closet $(P, [\cdot])$, a family $\mathrsfs{D}$ of subsets is \textit{union-complete} if, for every family $(D_x)_{x \in D}$ of elements of $\mathrsfs{D}$ indexed by some subset $D \in \mathrsfs{D}$ and such that $x \leqslant y \Rightarrow D_x \subset D_y$, we have $\bigcup_{x \in D} D_x \in \mathrsfs{D}$.  

\begin{theorem}\label{thm:interp}
Let $(P, [\cdot], \cc)$ be a continuous enriched closet whose family of weakly-closed subsets is union-complete. 
Then $P$ is strongly continuous, if and only if every subset is relatively-closed, i.e.\ $\cc$ is a closure operator. 
\end{theorem}

\begin{proof}
Assume that $\ll$ is interpolating, and let us show that every 
subset $A$ is relatively-closed. 
Let $x \in \cc(\cc(A)))$. 
Then $\twoheaddownarrow x \subset [\cc(A)] = \cc(A)$. 
This implies $\twoheaddownarrow \twoheaddownarrow x \subset [A]$. 
Since $\ll$ is interpolating, we have $\twoheaddownarrow \twoheaddownarrow x = \twoheaddownarrow x$. 
Thus, $\twoheaddownarrow x \subset [A]$. 
Since $P$ is continuous, we obtain $x \in \cc(\twoheaddownarrow x) \subset \cc([A]) = \cc(A)$.  
We have shown that $\cc(\cc(A)) \subset \cc(A)$, which proves that $A$ is relatively-closed. 

Conversely, assume that every 
subset is relatively-closed. 
Let $x \in P$ and let $A = \twoheaddownarrow \twoheaddownarrow y$. 
Then $A = \bigcup_{y \in \twoheaddownarrow x} \twoheaddownarrow y$. Since subsets $\twoheaddownarrow z = \bigcap_{A : \cc(A) \ni z} [A]$ are weakly-closed as intersections of weakly-closed subsets, and since the family of weakly-closed subsets is union-complete by assumption, this entails that $A$ is weakly-closed. 
Moreover $A$ is relatively-closed by assumption. 
So $\ll$ is interpolating by Lemma~\ref{lem:interp1}. 
\end{proof}

\begin{corollary}\label{cor:strongly}
Let $(P, [\cdot], \cc)$ be an enriched closet whose family of weakly-closed subsets is union-complete. 
Then $P$ is strongly continuous if and only if $\cc$ is a closure operator that preserves arbitrary intersections of weakly-closed subsets. 
\end{corollary}

\begin{proof}
Combine the previous result with Theorem~\ref{thm:global}. 
%
\end{proof}

\begin{example}[Example~\ref{ex:c4} continued]\label{ex:c5}
Let $(P, \leqslant)$ be a complete lattice, and $c \colon P \rightarrow P$ be an order-preserving inflationary map. 
Then $(P, \leqslant, \cc_c)$ is strongly continuous if and only if $(k, c)$ is a Galois connection, $k(x) \ll x$ for all $x \in P$, and $c \circ c = c$, where $k(x) := \bigvee \twoheaddownarrow x$. 
\end{example}

\begin{proposition}\label{prop:scott2}
Let $(P, [\cdot], \cc)$ be an enriched closet whose family of weakly-closed subsets is union-complete, and suppose that $\cc$ is a closure operator. 
Then $P$ is (strongly) continuous if and only if every open subset $G$ is way-upper, i.e.\ $G = \twoheaduparrow G$. 
\end{proposition}

\begin{proof}
By Proposition~\ref{prop:scott} we already know that, if $P$ is continuous, then $G = \twoheaduparrow G$ for all open subsets $G$. 
Conversely, assume that $\cc$ is a closure operator and $G = \twoheaduparrow G$ for all open subsets $G$, and let us show that $P$ is continuous. 
So let $x \in P$. 
Since $\cc$ is a closure operator, the subset $G = P \setminus \cc(\twoheaddownarrow x)$ is open. 
This implies that $G = \twoheaduparrow G$. 
So if we suppose that $x \notin \cc(\twoheaddownarrow x)$, then $x \in G$, so there exists some $y \in G$ such that $y \ll x$. 
But $y \ll x$ gives $y \notin G$ by definition of $G$,  a contradiction. 
Consequently, we have shown that $x \in \cc(\twoheaddownarrow x)$ for all $x$, i.e.\ that $P$ is continuous (so $P$ is strongly continuous by Corollary~\ref{cor:strongly}). 
\end{proof}

The following corollary extends Ern\'{e} and Wilke \cite[Theorem~2.3]{Erne83}. It is not clear whether a converse statement holds. 

\begin{corollary}
If $(P, [\cdot], \cc)$ is a strongly continuous enriched closet, then the family of closed subsets is completely distributive.  
\end{corollary}

\begin{proof}
Since $P$ is strongly continuous, $\cc$ is a closure operator that preserves arbitrary intersections of weakly-closed subsets. 
So the map $A \mapsto \cc(A)$ from the family of weakly-closed subsets to the family of closed subsets is a surjective morphism of complete lattices. Hence by Raney \cite[Theorem~1]{Raney52}, the family of closed subsets is completely distributive. 
\end{proof}

\section{Inner-regular enriched closets}\label{sec:inner-regular}

In this section we define the notion of preclosure space generated by a family of subsets. 
We shall see that deep links exist with the notions of continuity and strong continuity introduced in the previous sections. 
	
\begin{definition}
An enriched closet $(P, [\cdot], \cc)$, or the preclosure operator $\cc$, is \textit{generated by a family $\mathrsfs{D}$} of subsets of $P$ if for all subsets $A$ we have 
\[
\cc(A) = \bigcup_{D \subset [A]} \cc(D), 
\]
where the union is taken over the subsets $D \in \mathrsfs{D}$ included in $[A]$. 
In the case where $\cc$ is generated by the family of relatively-closed subsets, we say that $P$ is \textit{inner-regular}. 
\end{definition}

As an example, a continuous enriched closet is generated by the family made up of the subsets of the form $\twoheaddownarrow x$, and a strongly continuous enriched closet is inner-regular by Theorem~\ref{thm:interp}. 

\begin{lemma}
Let $(P, [\cdot], \cc)$ be an inner-regular enriched closet. 
Then singletons and 
principal ideals in $P$ are relatively-closed. 
\end{lemma}

\begin{proof}
Let $x \in P$.  
By inner-regularity there is some relatively-closed subset $D \subset \downarrow\!\! x$ such that $x \in \cc(D)$. This implies that $\cc(D) \subset \cc(\downarrow\!\! x) \subset \cc(\cc(D)) = \cc(D)$, so that $\cc(\downarrow\!\! x) = \cc(D)$ is closed, i.e.\ $\downarrow\!\! x$ is relatively-closed. 
Since $\cc(\{x\}) = \cc(\downarrow\!\! x)$ by compatibility of $\cc$ with $[\cdot]$, we also get that $\{ x\}$ is relatively-closed.  
\end{proof}


\begin{proposition}\label{prop:twohead}
Let $(P, [\cdot], \cc)$ be a continuous enriched closet, and let $\mathrsfs{D}$ be a family of subsets of $P$. 
Then $\cc$ is generated by $\mathrsfs{D}$ if and only if every subset of the form $\twoheaddownarrow x$ can be written as $[D]$ for some $D \in \mathrsfs{D}$. 
In particular, $P$ is inner-regular if and only if every subset of the form $\twoheaddownarrow x$ is relatively-closed. 
\end{proposition}

\begin{proof}
First assume that $\cc$ is generated by $\mathrsfs{D}$, and let $x \in P$. By continuity we have $x \in \cc(\twoheaddownarrow x)$, so there exists some  subset $D \in \mathrsfs{D}$ such that $D \subset [\twoheaddownarrow x] = \twoheaddownarrow x$ and $x \in \cc(D)$. The latter implies $\twoheaddownarrow x \subset [D]$. So we obtain $\twoheaddownarrow x = [D]$. 

Conversely, assume that every subset of the form $\twoheaddownarrow x$ can be written as $[D]$ for some $D \in \mathrsfs{D}$, and let $A$ be a subset of $P$. If $x \in \cc(A)$, then $\twoheaddownarrow x \subset [A]$, so $D \subset [A]$, where $D \in \mathrsfs{D}$ is such that $[D] = \twoheaddownarrow x$. 
Moreover, $x \in \cc(\twoheaddownarrow x) = \cc([D]) = \cc(D)$. 
This proves that $\cc(A)$ is included in $\bigcup_{D \subset [A]} \cc(D)$ for every subset $A$, i.e.\ that $\cc$ is generated by $\mathrsfs{D}$. 
\end{proof}

The following result generalizes Novak \cite[Proposition~1.22]{Novak82b}. 

\begin{theorem}
Let $(P, [\cdot], \cc)$ be a continuous enriched closet. 
Then $P$ is strongly continuous if and only if $\cc$ is generated by a union-complete family of relatively-closed weakly-closed subsets (in particular $P$ is inner-regular).  
\end{theorem}


\begin{proof}
Suppose first that $P$ is strongly continuous. 
Then it is not difficult to prove that $\cc$ is generated by the family $\{ \twoheaddownarrow x : x \in P \}$, which is indeed union-complete by the interpolation property, and made of relatively-closed weakly-closed subsets. 

Conversely, let $\mathrsfs{I}$ be a union-complete family of relatively-closed and weakly-closed subsets that generates $\cc$, and let $x \in P$. 
For every $y \in P$ we can write $\twoheaddownarrow y = [I_y]$ for some $I_y \in \mathrsfs{I}$, thanks to Proposition~\ref{prop:twohead}. 
Since every element of $\mathrsfs{I}$ is weakly-closed, this shows that $\twoheaddownarrow y = I_y \in \mathrsfs{I}$. 
Since $\mathrsfs{I}$ is union-complete, we have $I_x' := \bigcup_{y \in I_x} I_y \in \mathrsfs{I}$. 
This implies that $I_x'$ is weakly-closed and relatively-closed. 
Moreover, $I_x' = \bigcup_{y \in I_x} I_y = \bigcup_{y \in I_x} \twoheaddownarrow y = \twoheaddownarrow I_x$. 
Since $\twoheaddownarrow x = I_x$, we get $\twoheaddownarrow \twoheaddownarrow x = \twoheaddownarrow I_x = I_x'$. 
Consequently, $\twoheaddownarrow \twoheaddownarrow x$ is weakly-closed and relatively-closed, for all $x \in P$. 
Now it follows from Lemma~\ref{lem:interp1} that $\ll$ is interpolating, i.e.\ that $P$ is strongly continuous. 
\end{proof}

\begin{example}[Example~\ref{ex:p4} continued]\label{ex:p5}
Let $(P, \leqslant)$ be a poset, let $\mathrsfs{D}$ be the family of directed subsets with supremum, and let $\mathrsfs{I}$ be the family of elements of $\mathrsfs{D}$ that are lower. 
Since $\cc_{\mathrsfs{D}}(I) = \downarrow\!\! \bigvee I$ for every $I \in \mathrsfs{I}$, the elements of $\mathrsfs{I}$ are relatively-closed. In particular, $P$ is inner-regular. Moreover, $\mathrsfs{I}$ is union-complete, since a directed union of directed subsets is directed. 
So if $(P, \leqslant, \cc_{\mathrsfs{D}})$ is continuous, i.e.\ if every $x \in P$ is the directed supremum of elements way-below it, 
then $P$ is strongly continuous by the previous proposition. 
Hence we recover a well-known result of domain theory (see e.g.\ \cite[Theorem~I-1.9]{Gierz03}). 
\end{example}

\begin{proposition}\label{prop:contD}
Let $(P, [\cdot], \cc)$ and $(P', [\cdot], \cc')$ be enriched closets. 
Assume that $\cc$ is generated by a family $\mathrsfs{D}$. 
Then a $[\cdot]$-continuous map $f \colon P \rightarrow P'$ is strictly-continuous if and only if 
$
f(\cc(D)) \subset \cc'(f(D))
$, 
for all $D \in \mathrsfs{D}$. 
\end{proposition}

\begin{proof}
Assume that $f$ satisfies the property of the proposition, and let $A$ be a subset of $P$. 
We show that $f(\cc(A)) \subset \cc'(f(A))$. 
So let $y \in f(\cc(A))$. There is some $x \in \cc(A)$ with $y = f(x)$. 
Since the family $\mathrsfs{D}$ generates $\cc$, there is some $D \in \mathrsfs{D}$ such that $D \subset [A]$ and $x \in \cc(D)$. 
By hypothesis we have $f(\cc(D)) \subset \cc'(f(D))$, so that $y \in \cc'(f(D))$. 
Moreover, $f$ is $[\cdot]$-continuous, so with the inclusion $D \subset [A]$ we obtain $f(D) \subset [f(A)]$, so that $\cc'(f(D)) \subset \cc'([f(A)]) = \cc'(f(A))$. 
This shows that $y \in \cc'(f(A))$ for all $y \in f(\cc(A))$, i.e.\ that $f(\cc(A)) \subset \cc'(f(A))$, which is the desired result. 
\end{proof}


%

\begin{corollary}
Let $(P, [\cdot], \cc)$ and $(P', [\cdot], \cc')$ be enriched closets. 
Assume that $P$ is inner-regular. 
If a $[\cdot]$-continuous map $f \colon P \rightarrow P'$ preserves relatively-closed subsets, 
then $f$ is strictly-continuous if and only if $f$ is closure-continuous if and only if 
$
f(\cc(D)) \subset \cc'(f(D))
$, 
for all relatively-closed subsets $D$. 
\end{corollary}

\begin{proof}
Assume that $f$ is closure-continuous, and let $D$ be a relatively-closed subset of $P$. 
Since $f$ is closure-continuous we have 
\[
f(\cc(D)) \subset f(\overline{\cc}(D)) \subset \overline{\cc}'(f(D)). 
\]
But $f$ preserves relatively-closed subsets, so $f(D)$ is relatively-closed, i.e.\ $\overline{\cc}'(f(D)) = \cc'(f(D))$. 
It shows that $f(\cc(D)) \subset \cc'(f(D))$ for all relatively-closed subsets $D$. 
So by Proposition~\ref{prop:contD} $f$ is strictly-continuous. 
\end{proof}


The following result focuses on closure-continuity. 
We can avoid assuming that $f$ be $[\cdot]$-continuous at the price of a few extra hypotheses. 
Especially, we shall suppose that $\cc$ and $[\cdot]$ are \textit{jointly generated} by a family $\mathrsfs{D}$, in the sense that $\cc$ is generated by $\mathrsfs{D}$ and 
\[
[A] = \bigcup_{D \in \mathrsfs{D}, D \subset A} [D], 
\]
for all subsets $A$. 

\begin{proposition}\label{prop:contDD1}
Let $(P, [\cdot], \cc)$ and $(P', [\cdot], \cc')$ be enriched closets, and let $f \colon P \rightarrow P'$. 
Assume that $\cc$ and $[\cdot]$ are jointly generated by a family $\mathrsfs{D}$ such that $f(D)$ is relatively-closed for all $D \in \mathrsfs{D}$. 
Then $f$ is closure-continuous if and only if 
$
f(\cc(D)) \subset \cc'(f(D))
$, 
for all $D \in \mathrsfs{D}$. 
\end{proposition}

\begin{proof}
Suppose that $f(\cc(D)) \subset \cc'(f(D))$ for all $D \in \mathrsfs{D}$. 
To prove that $f$ is closure-continuous, we prove that $f^{-1}(F')$ is closed for every closed subset $F'$. 
So let $A \subset f^{-1}(F')$, and let us show that $\cc(A) \subset f^{-1}(F')$. 
Let $x \in \cc(A)$.
Since $\cc$ is generated by $\mathrsfs{D}$ there is some $D_0 \in \mathrsfs{D}$ such that $D_0 \subset [A]$ and $x \in \cc(D_0)$. 
Then $f(x) \in f(\cc(D_0)) \subset \cc'(f(D_0))$. 
Since $D_0 \subset [A]$, we have $f(D_0) \subset \bigcup_{D \subset A} f([D])$, where $D$ runs over the subsets of $A$ that belong to $\mathrsfs{D}$. 
Moreover, if $D \in \mathrsfs{D}$ and $D \subset A$, then $[D] \subset \cc(D)$, so $f([D]) \subset f(\cc(D)) \subset \cc'(f(D)) \subset \cc'(f(A)) \subset \cc'(F') = F'$. 
This shows that $f(D_0) \subset F'$, so $f(x) \in \cc'(f(D_0)) \subset \cc'(F') = F'$. 
So $x \in f^{-1}(F')$. This proves that $f^{-1}(F')$ is closed. 

Conversely, suppose that $f$ is closure-continuous, and  let $D \in \mathrsfs{D}$. 
Then 
$
f(\cc(D)) \subset f(\overline{\cc}(D)) \subset \overline{\cc}'(f(D)) = \cc'(f(D))
$, 
by relative closedness of $f(D)$. 
\end{proof}

\section{Topological preclosure operators}\label{sec:topology}

A preclosure space $(E, \cc)$ is called a \textit{topological preclosure space} if the open sets form a topology. In this case $\cc$ is called a \textit{topological preclosure operator}. 
An equivalent condition is that the closure operator $\overline{\cc}$ associated with $\cc$ satisfy $\overline{\cc}(\emptyset) = \emptyset$ together with the property 
\[
\overline{\cc}(A \cup B) \subset \overline{\cc}(A) \cup \overline{\cc}(B), 
\]
for all subsets $A, B$ of $E$, i.e.\ that $\overline{\cc}$ be a \textit{Kuratowski closure operator}. 

A subset $R$ of a preclosure space is \textit{irreducible} if $R$ is nonempty and 
\[
R \subset F \cup F' \Rightarrow R \subset F \mbox{ or } R \subset F', 
\]
for all closed subsets $F, F'$. 
Note that a directed union of irreducible subsets is irreducible, and $\cc(R)$ is irreducible whenever $R$ is irreducible. 

\begin{example}
Let $(P, \leqslant, \cc)$ be an enriched qoset where $\cc$ is the Alexandrov closure operator $A \mapsto \downarrow\!\! A$ (see Example~\ref{ex:alex}). 
Then the irreducible subsets coincide with the directed subsets. 
\end{example}


\begin{proposition}
Let $(P, [\cdot], \cc)$ be an enriched closet. 
Assume that $[\cdot]$ is topological and that $\cc$ is generated by irreducible subsets. 
Then $\cc$ is topological.  
\end{proposition}

\begin{proof}
Let $F$ and $F'$ be closed subsets of $P$, and let $x \in \cc(F \cup F')$. 
Then $x \in \cc(R)$ for some irreducible subset $R \subset [F \cup F']$. 
Since $[\cdot]$ is topological by assumption, we have $[F \cup F'] = [F] \cup [F'] \subset \cc(F) \cup \cc(F') = F \cup F'$, so $R \subset F \cup F'$. 
Now $R$ is irreducible, so we can suppose that $R \subset F$ without loss of generality. 
Thus, $x \in \cc(R) \subset \cc(F) = F$. 
This proves that the union of two closed subsets is closed, or that the intersection of two open subsets is open. 
We already know that open subsets are stable by arbitrary unions, and that $\emptyset$ is open. So to prove that the open subsets form a topology, it remains to show that $P$ is open, i.e.\ that $\cc(\emptyset) = \emptyset$. This follows from the fact that $\cc$ is generated by irreducible subsets, and that irreducible subsets are nonempty by definition. 
\end{proof}



\begin{example}[Example~\ref{ex:p5} continued]\label{ex:p6}
Let $(P, \leqslant)$ be a poset, and let $\mathrsfs{D}$ be the family of directed subsets with supremum. 
We show that the preclosure $\cc_{\mathrsfs{D}}$ is generated by irreducible subsets. 
For this it suffices to prove that every directed subset $D$ is irreducible. 
Indeed, suppose that $D \subset F \cup F'$ with $F, F'$ closed. If $D \not\subset F$ and $D \not\subset F'$, then there are some $x \in D\setminus F$ and $x' \in D\setminus F'$. 
The subset $D$ is directed so there exists some $x'' \in D$ such that $x \leqslant x''$ and $x' \leqslant x''$. 
Since $D \subset F \cup F'$ we have $x'' \in F$ or $x'' \in F'$, we get $x \in \downarrow\!\! F$ or $x' \in \downarrow\!\! F'$. 
Closed subsets $F$ and $F'$ are lower, so $x \in F$ or $x' \in F'$, a contradiction. 
We have shown that $\cc_{\mathrsfs{D}}$ is topological (the topology made up of the open subsets is known as the \textit{Scott topology}). 
\end{example}

\begin{example}[Example~\ref{ex:c5} continued]\label{ex:c6}
Let $(P, \leqslant)$ be a qoset, and $c \colon P \rightarrow P$ be an order-preserving and inflationary map. 
The preclosure $\cc_{c}$ is generated by singletons, which are irreducible subsets, hence is topological. 
More than that, in this situation every union of closed subsets is closed. 
\end{example}


\begin{proposition}
Let $(P, [\cdot], \cc)$ be a continuous enriched closet. 
If $\cc$ is topological then $\cc$ is generated by irreducible subsets. 
\end{proposition}

\begin{proof}
We first show that every subset of the form $\twoheaddownarrow x$ is irreducible. 
Note that $\twoheaddownarrow x$ is nonempty: if $\twoheaddownarrow x = \emptyset$ then $x \in \cc(\emptyset) = \emptyset$, a contradiction. 
Now assume that $\twoheaddownarrow x \subset F \cup F'$ for some closed subsets $F, F'$. 
Since $P$ is continuous and $\cc$ is topological we have $x \in \cc(\twoheaddownarrow x) \subset F \cup F'$. 
Suppose without loss of generality that $x \in F$. 
Then $x \in \cc(F)$, so $\twoheaddownarrow x \subset [F] \subset \cc(F) = F$. 
This shows that $\twoheaddownarrow x$ is irreducible. 

Once one has seen that the continuity of $P$ implies that $\cc$ is generated by the subsets of the form $\twoheaddownarrow x$, the proof is complete. 
\end{proof}

\begin{acknowledgements}
I would like to gratefully thank Marcel Ern\'e for his encouragements and comments on a previous version of the manuscript. 
\end{acknowledgements}

\bibliographystyle{plain}

\def\cprime{$'$} \def\cprime{$'$} \def\cprime{$'$} \def\cprime{$'$}
  \def\ocirc#1{\ifmmode\setbox0=\hbox{$#1$}\dimen0=\ht0 \advance\dimen0
  by1pt\rlap{\hbox to\wd0{\hss\raise\dimen0
  \hbox{\hskip.2em$\scriptscriptstyle\circ$}\hss}}#1\else {\accent"17 #1}\fi}
  \def\ocirc#1{\ifmmode\setbox0=\hbox{$#1$}\dimen0=\ht0 \advance\dimen0
  by1pt\rlap{\hbox to\wd0{\hss\raise\dimen0
  \hbox{\hskip.2em$\scriptscriptstyle\circ$}\hss}}#1\else {\accent"17 #1}\fi}

\end{document}